\documentclass[11pt]{article}
\usepackage{amsmath,amssymb,epsfig,color,url,times}
\usepackage{comment}
\usepackage[inline,shortlabels]{enumitem}
\usepackage{authblk}
\usepackage{bm}
\usepackage{environ}
\usepackage{tabto}
\usepackage[margin=1in]{geometry}
\usepackage{amsthm}

\usepackage{algorithm, algpseudocode}
\algrenewcommand\algorithmicrequire{\textbf{\quad Input:}}
\algrenewcommand\algorithmicensure{\textbf{\quad Output:}}

\newcommand{\qedsymb}{\hfill{\rule{2mm}{2mm}}}
\renewenvironment{proof}{\begin{trivlist} \item[\hspace{\labelsep}{\bf \noindent Proof.\/}] }{\qedsymb\end{trivlist}}%
\newcommand{\moo}{\mathbf{Div\_Alloc}}
\newcommand{\mooo}{\mathbf{Indiv\_Alloc}}
\newcommand{\mooa}{\mathbf{Sell\_Div}}
\newcommand{\moob}{\mathbf{Sell\_Modif}}
\newcommand{\moov}{\mathbf{MVCG}}
\newcommand{\moos}{\mathbf{RS\_Liquid}}
\newcommand{\mool}{\mathbf{Liquid\_Div}}
\newcommand{\mo}{\mathbf{Rev\_Indiv}}
\newcommand{\modd}{\mathbf{Rev\_Div}}
\newcommand{\mof}{\mathbf{Rev\_Offline}}
\newcommand{\mtempl}{\mathbf{RS\_Online}}
\newcommand{\alloc}{Alloc}
\newcommand{\pay}{\mathcal{P}}
\newcommand{\payrev}{P^m_{rev}}
\newcommand{\mytab}{$\mbox{   }$$\mbox{   }$$\mbox{   }$$\mbox{   }$}
\newcommand{\lp}{\left(}
\newcommand{\rp}{\right)}
\newcommand{\lb}{\left[}
\newcommand{\rb}{\right]}

\newcommand{\argmax}{\mathrm{argmax}}

\theoremstyle{plain}
\newtheorem{theorem}{Theorem}
\newtheorem{lemma}{Lemma}

\newtheorem{observation}{Observation}

\usepackage[colorlinks,
linkcolor=blue,
filecolor=blue,
citecolor=magenta,
urlcolor=blue,
pdfstartview=FitH]{hyperref}
\DeclareMathOperator{\E}{E}

\newcommand{\MyFrame}[1]{\noindent \framebox[\textwidth]{ \begin{minipage}{0.97\textwidth} #1 \end{minipage}}}%

\newif\ifshowproof
\showprooffalse

\NewEnviron{Proof}{%
	\ifshowproof%
	\begin{proof}%
		\BODY
	\end{proof}%
	\fi%
}%

\begin{document}
\title{Truthful Secretaries with Budgets}

\author[1]{Alon Eden}
\author[1]{Michal Feldman}
\author[1]{Adi Vardi}
\affil[1]{Blavatnik School of Computer Science, Tel-Aviv University}
\affil[ ]{\textit{alonarden@gmail.com, michal.feldman@cs.tau.ac.il, adivardi@gmail.com} }
\renewcommand\Authands{ and }

\maketitle

\vspace{-0.1cm}
\begin{abstract}
We study online auction settings in which agents arrive and depart dynamically in a random (secretary) order, and each agent's private type consists of the agent's arrival and departure times, value and budget.
We consider multi-unit auctions with additive agents for the allocation of both divisible and indivisible items.
For both settings, we devise truthful mechanisms that give a constant approximation with respect to the auctioneer's revenue, under a large market assumption.
For divisible items, we devise in addition a truthful mechanism that gives a constant approximation with respect to the liquid welfare --- a natural efficiency measure for budgeted settings introduced by Dobzinski and Paes Leme [ICALP'14].
Our techniques provide high-level principles for transforming offline truthful mechanisms into online ones, with or without budget constraints.
To the best of our knowledge, this is the first work that addresses the non-trivial challenge of combining online settings with budgeted agents.
\end{abstract}

\section{Introduction}
In a typical setting of sales of online ad slots, advertisers arrive
at different times, each with her own preferences and budget. The
auctioneer (e.g., \textit{cnn.com}) has to make decisions about the allocation of the ad slots to the advertisers and how much to
charge them. Such settings give rise to various
optimization problems associated with different objective functions.
The most common objective functions that have been considered in the
literature are \begin{enumerate*}[(a)]
	\item social welfare: the sum of the bidders' valuations from their allocations, and
	\item the auctioneer's revenue: the total payment the auctioneer derives from the sale.
\end{enumerate*}

This scenario inspired the online model studied by Hajiaghayi et al.~\cite{hajiaghayi2004adaptive} , who used the {\em random sampling} paradigm to design a truthful mechanism for this setting.
The random sampling paradigm was first introduced in the context of selling $m$ identical items in offline settings without budgets \cite{goldberg2001competitive}. This mechanism divides the agents into two sets, $A$ and $B$, by tossing a fair coin. Next, it calculates prices $p_A$ and $p_B$ (per item) based on the sets $A$ and $B$, respectively.
Finally, it sells $m/2$ items at price $p_A$ to agents in $B$, in an arbitrary order, and similarly $m/2$ items at price $p_B$
to agents in $A$. This mechanism is known to give a constant approximation with respect to the optimal revenue.

Hajiaghayi et al.~\cite{hajiaghayi2004adaptive} managed to apply the random sampling paradigm to the online setting, while losing only a constant factor in the approximation of the revenue.
Their mechanism is based on the following observations and techniques:
\begin{itemize}
\item A random partition of the $n$ agents into two sets in the offline setting is equivalent to the following procedure in the online setting:
Given a set of $n$ agents arriving in a random order, place the first $j$ agents \big(where $j$ is sampled from the binomial distribution $B(n,1/2)$\big) in one set, called the sampling set, and the remaining ones in a second set, called the performance set.
\item One can acquire the private types of agents in the sampling set by applying the VCG mechanism on this set.
\item Since the random sampling mechanism sells items in an arbitrary order, it is without loss of generality that items are sold to
agents in the performance set according to their online order of arrival.
\end{itemize}

Applying the random sampling paradigm to online settings was a big step toward the applicability of truthful mechanisms. Yet, like many other papers in the mechanism design literature, it ignored the fact that agents are usually restricted by budget constraints. Since budget constraints have a huge effect on agent behavior in real-life scenarios, our goal in this paper is to incorporate budgets into the online setting described above.

The allocation of indivisible items  among agents with budget constraints in an offline setting was considered by Borgs et al.~\cite{borgs2005multi}. As identified by~\cite{borgs2005multi}, budget constraints impose many challenges on auction settings.
For example, the seminal VCG mechanism
~\cite{groves1973incentives} 
 fails since the utility function is no longer quasi-linear.
Borgs et al.~\cite{borgs2005multi} devised a truthful mechanism that adjusted the random sampling mechanism to deal with budget constraints.
Their mechanism approximates the revenue within a constant factor under a large market assumption.

Budget constraints impose challenges also with respect to the social welfare objective function. It is well known that if budgets are private, no truthful mechanism can give a non-trivial approximation (better than $n$) to the social welfare, even for the case of a single item.
This is because any such mechanism must allocate the item to an agent with a high value relative to all other agents,
for any budget she reports.
To make things worse, even if budgets are public, it has been shown in \cite{Dobzinski2013Liquid} that
no truthful auction among $n$ agents can achieve a better than $n$ approximation to the optimal social welfare.

This problem was recently addressed by Dobzinski and Paes Leme.~\cite{Dobzinski2013Liquid}, who introduced the {\sl liquid welfare}
objective function. A natural interpretation of the social welfare
function is the maximum revenue an omniscient auctioneer can
extract. Applying this interpretation to budgeted settings gives rise to the {\em liquid welfare} objective function,
defined as follows:
For every agent, consider the minimum between her value for the allocation (i.e., her willingness to pay) and her budget (i.e., her ability to pay), and take the sum of this minimum over all agents.
Dobzinski and Paes Leme considered the allocation of a divisible item to agents with public budgets (and private values), and devised truthful mechanisms that give a constant approximation to the optimal liquid welfare.
Lu et al. \cite{lu2014liquid} extended this to agents with {\em private} budgets and any monotonically non decreasing valuation function. The mechanism devised in \cite{lu2014liquid} is a probabilistic combination of a random sampling mechanism and a modified VCG mechanism.

In the present paper, the main challenge is the combination of two non-trivial challenges; namely the online setting  and budgeted agents.
From the description above it seems that a lot of the necessary components are already in place.
First, ~\cite{borgs2005multi} and \cite{lu2014liquid} devised random sampling based mechanisms that give a constant approximation
with respect to revenue and liquid welfare in budgeted offline settings. Second, ~\cite{hajiaghayi2004adaptive} presented a
technique that adjusts mechanisms based on the random sampling paradigm to online settings.
A natural approach for dealing with online budgeted settings would be to combine these techniques to
obtain  nearly-optimal mechanisms. Not surprisingly, as demonstrated in the sequel,
this approach fails badly.
Indeed, in order to address the combination of online settings and budget constraints, new ideas are required.

\subsection{Our Model} \label{Setting}
We consider the combination of the online setting without budgets
of \cite{hajiaghayi2004adaptive} and the offline settings with
budgets of ~\cite{lu2014liquid,borgs2005multi}. In our setting, there is
a set $N=\{1,\ldots, n\}$ of agents. Every agent $i\in N$ has a
private type (known only to the agent) that is represented by a
tuple $\lp a_i, d_i, v_i, b_i \rp$, where:
\begin{itemize}
    \item $a_i$ and $d_i$ are the respective arrival and departure time of agent $i$ (clearly, $d_i\geq a_i$).
    The interval $[a_i,d_i]$ is referred to as agent $i$'s \emph{time frame}.
    \item $v_i$ is agent $i$'s value for receiving a single item within her time frame.
    \item $b_i$  is the budget of agent $i$.
\end{itemize}

We consider an online auction with a secretary ``flavor." An
adversary states a vector of time frames $\lp\lb a_1, d_1\rb,\lb
a_2,d_2\rb,\ldots,\lb a_n, d_n\rb\rp$ such that $a_i< a_{j}$ for
every $i< j$, and a vector of pairs (value, budget), $(\lb v_1,
b_1\rb,$ $\lb v_2,b_2\rb,\ldots,\lb v_n, b_n\rb)$. A random
permutation is used to match these pairs and determine the types of
the agents. More formally, let $\Pi_N$ be the set of all possible
permutations on the set $N$. A permutation $\pi:N\mapsto N$ is sampled
uniformly at random from the set $\Pi_N$. Then, agent $i$'s type is given by the tuple $\lp a_i, d_i, v_{\pi\lp i\rp},
b_{\pi\lp i\rp}\rp$.

As in standard secretary-like problems, a random order is essential
to get any guarantee on the performance of the mechanism. As in
\cite{hajiaghayi2004adaptive} \footnote{Based on personal
communication with the authors, this is essentially what is assumed
for the correctness of Mechanism $RM_k$ in Section 6 in \cite{hajiaghayi2004adaptive}.}, we assume
that arrival times are distinct, but the results also extend to
non-distinct arrival times if agents cannot bid before their real
arrival times \footnote{We refer to Appendix \ref{appsec:tie-break} for a description of the tie-breaking rule for this case.}. It should be noted that most of the papers studying online settings with strategic agents
consider restrictive inputs or weaker solution concepts (e.g.,
~\cite{lavi2000competitive,friedman2003pricing,hajiaghayi2004adaptive,hajiaghayi2005online,lavi2005online,parkes2003mdp}).

In our model, agents can manipulate any component of their type. In particular, they can report earlier or later arrival and departure times, and arbitrary values and budgets.
Upon arrival, agents report their type. Based on these reports, the mechanism determines an allocation and a payment for each agent.
The utility of agent
$i$ for obtaining an allocation of $x_i$, within her {\em real} time frame, for a payment of $p_i$ is:
\begin{eqnarray} \label{ValuationFunction}
u_i(x_i,p_i)=
\begin{cases}
v_i  x_i - p_i &p_i\leq b_i\\
-\infty &p_i>b_i\\
\end{cases}.
\end{eqnarray}

We assume that agents are risk neutral.
We consider mechanisms that satisfy the following properties:
\begin{itemize}
    \item Feasibility: The mechanism does not sell more items than are available.
    \item Individual Rationality: An agent's expected utility from an allocation and payment is non-negative.
    \item Incentive Compatibility: An agent's expected utility is maximized when she reports her true type.
\end{itemize}
Since all the mechanisms devised in this paper are trivially feasible and individually rational, our analysis focuses on incentive compatibility.
\subsection{Our Results and Techniques}  \label{result}

In Section \ref{sec:highlevel} we establish techniques that constitute high level principles for transforming nearly optimal offline mechanisms into online ones.

In Section \ref{sec:div_indiv} we present methods for the allocation of divisible and indivisible items to budgeted agents who arrive online. We first describe a simple greedy allocation method for divisible items at a given price per item $p$, where agents are ordered randomly. Every agent whose value exceeds $p$ is allocated as many items as she can afford with her budget at a price per item $p$, limited by the number of remaining items.

In an offline setting, this method can be used for selling indivisible items as well. This is done by translating fractions of items into lottery tickets with the corresponding probabilities, and resolving the market via correlated lotteries upon the "arrival" of the last agent. In an online setting, this method cannot be applied since by the time the last agent arrives, all previous agents may have already departed. To overcome this obstacle, we present an allocation method that resolves the agent's lottery in an online fashion, and loses only a factor 2 in revenue compared to the offline setting.

In Section \ref{sec:rs_base} we present a prototype mechanism that transforms an offline mechanism based on the random sampling paradigm into an online one. The prototype mechanism receives as input the set of agents, the number of items to be allocated, a pricing function and an allocation function.

Our starting point is the mechanism devised in \cite{hajiaghayi2004adaptive} for an online setting without budget constraints. This mechanism splits the agents into a sampling set and a performance set.
It induces agents in the sampling set to reveal their true values by applying VCG on them, and uses this information to extract revenue from the agents in the performance set.
In many settings (such as budgeted settings), VCG cannot be applied. To mitigate this problem, the prototype mechanism applies a random sampling based method on the sampling set.

The above technique ensures that in cases where arrival time is public information, agents in the sampling set do not have incentive to misreport their private information (such as the value and budget in our model). Since we consider models in which arrival time is also private, we need to ensure that agents have no incentive to manipulate their arrival time either. For agents with unit-demand valuations, this issue is handled in \cite{hajiaghayi2004adaptive} by ensuring that the price used in the sampling set never exceeds that in the performance set, and every agent in the sampling set whose value exceeds the price receives it.

This technique, however, does not apply to agents with additive valuations.
Indeed, there are cases where the price used in the performance set is higher than the price used in the sampling set, and yet an agent in the sampling set may wish to delay her arrival time, in order to receive a larger number of items (as part of the performance set).  In order to ensure an agent cannot gain by delaying her arrival time, we treat an agent in the performance set \emph{as if} she was in the sampling set. In particular, an agent in the performance set receives at most the allocation she would have received had she been in the sampling set, and the same price.

We demonstrate the broad applicability of this prototype by applying it to settings with budgeted agents for the allocation of two different types of items (namely, divisible and indivisible) and for two different objective functions (namely, revenue and liquid welfare maximization).

In Section \ref{sec:rev} we devise mechanisms for the allocation of identical items (both divisible and indivisible) that approximate the optimal revenue. This is done by applying the above prototype using appropriately chosen pricing and allocation functions.
As in \cite{borgs2005multi}, without a large market assumption,
no truthful mechanism can achieve a constant approximation to the optimal revenue, even for the case of a single item.
This is because an agent whose budget is significantly greater than that of the others must receive the item for any value she reports. We establish the following theorem:\\\\
{\bf Theorem:} There exist truthful mechanisms for the allocation of divisible and indivisible items among budgeted agents with additive valuations in online settings that give a constant approximation to the optimal revenue, under a large market assumption.\\\\
In Section \ref{sec:liquid} we devise a mechanism for the allocation of identical divisible items that approximates the optimal liquid welfare.\\\\
{\bf Theorem:} There exists a truthful mechanism for the allocation of divisible items among budgeted agents with additive valuations in online settings that gives a constant approximation to the optimal liquid welfare.\\\\
Under a large market assumption, a similar mechanism to the one presented in Section \ref{sec:rev}
approximates the optimal liquid welfare to within a constant factor.
For liquid welfare, though, we devise a new truthful mechanism that gives a constant approximation for any market (i.e., without employing large market assumption). As in the offline setting described in \cite{lu2014liquid},
our mechanism is a probabilistic combination of a random sampling based mechanism and a modified VCG mechanism. However, both mechanisms should be adjusted to the online setting.
The random sampling mechanism is adjusted using the prototype mechanism described in Section \ref{sec:rs_base}.

The modified VCG mechanism should also be adjusted, since in the original one described in \cite{lu2014liquid},
all bids are received simultaneously, and this is inherently infeasible in online settings.
Indeed, when the last agent arrives, all other agents may have already departed. Therefore, we propose a new modified VCG as follows: The mechanism places the first $n/2$ agents in set $A$ and the rest in set $B$. Upon the arrival of the last agent in $A$, it applies the modified VCG mechanism of \cite{lu2014liquid} on the agents in set $A$, and determines a threshold value, which is greater than the highest bid in $A$. With probability 1/2, all the items are allocated to the winner in $A$; otherwise, all the items are allocated to the first agent in $B$ that surpasses the threshold. 
\subsection{Related Work}
Online auctions have been the subject of a vast body of work.
Lavi and Nisan \cite{lavi2000competitive} introduced an online model in which agents have decreasing marginal valuations over identical indivisible items. While in this model the agent's value is private, her arrival time is public. A wide variety of additional online auction settings, such as digital goods \cite{blum2005near,bar2002incentive} and combinatorial auctions \cite{awerbuch2003reducing}, have been studied under the assumption of private values and public arrival times.

Similarly to our model and the model presented in \cite{hajiaghayi2004adaptive}, online settings in which agents arrive in a random order were also considered in \cite{babaioff2007matroids,kleinberg2005multiple,kesselheim2013optimal}.

Friedman and Parkes \cite{friedman2003pricing} considered the case where an agent's arrival time is also part of her private information, and thus can also be manipulated. Since then, additional auction settings were studied under the assumption that an agent's arrival and departure times are private \cite{hajiaghayi2005online,parkes2003mdp,hajiaghayi2004adaptive,lavi2005online,friedman2003pricing}. 

The reader is referred to \cite{parkes2007online} for an overview of online mechanism design.

Offline mechanisms for budgeted agents have long been considered, both in the context of revenue maximization and of welfare maximization. In the context of revenue maximization, Abrams \cite{abrams2006revenue} also considered the model of \cite{borgs2005multi} where indivisible and identical items are sold to additive agents with hard budget constrains. She established the relation between the revenue achieved by an optimal mechanism that is allowed to offer multiple prices and the one of an optimal mechanism who is only allowed to sell at a single price. Maximizing revenue was also considered in Bayesian settings \cite{chawla2011bayesian} and with the goal of approximating the optimal envy free revenue \cite{devanur2013prior,feldman2012revenue}.

Considering the welfare maximization objective, as stated before, it is in general impossible to approximate the sum of agents' welfare using a truthful mechanism. To overcome this obstacle, a more relaxed solution concept of Pareto efficiency was considered in \cite{bhattacharya2010incentive,colini2012multiple,dobzinski2012multi,fiat2011single,goel2012polyhedral,goel2013clinching}, and in \cite{devanur2013prior} they focus on approximating the welfare of the optimal envy free mechanism.  

The most related works to our model are \cite{hajiaghayi2004adaptive,borgs2005multi,Dobzinski2013Liquid,lu2014liquid}, which are discussed in great detail throughout the paper.

\section{High Level Techniques}\label{sec:highlevel}

\subsection{Divisible and Indivisible Allocation} \label{sec:div_indiv}
In this section we present methods for the allocation of divisible and indivisible items in online settings with budgets.
Let $S$ be a set of agents, $p$ a price and $k$ an amount of items.
Recall that every agent $i$ is associated with value $v_i$ and budget $b_i$.
We define $\moo$ to be a procedure for the allocation of divisible items:
\begin{figure}[H]
	\MyFrame{
		$\moo\lp S, p, k\rp$
		\begin{enumerate}
			\item Initialize $k'\gets k$, $x_i\gets 0$ for every $i\in S$.
			\item Let $\pi$ be a permutation of $S$ chosen uniformly at random from the set of all permutations.
			\item For $i=1,2,\ldots,|S|$:
			\begin{itemize}
				\item If $v_{\pi(i)}\geq p$:\\
					\mytab $x_{\pi(i)} \gets \min\{b_{\pi(i)}/p, k'\}$\\
					\mytab $k'\gets k'-x_{\pi(i)}$
			\end{itemize}
			\item Return $(x,x)$.
		\end{enumerate}
	} \caption{A procedure for the allocation of divisible items.}

\label{alg:M}
\end{figure}
$\moo$ is a greedy allocation procedure for divisible items at a given price per item $p$. 
Note that $\moo$ returns two allocation vectors: one for the allocation and a second one for purposes of charging. 
While in this procedure the two allocations are identical, we keep this format for the sake of uniformity across divisible and indivisible allocations (see below).

Let $\mathcal{M}$ be a mechanism that returns an allocation vector $x$, such that $x_i$ is the allocation of agent $i$.
We define $\E_{\mathcal{M}}[x_i]$ to be the expected value of $x_i$, where the
expectation is taken over the coin flips of $\mathcal{M}$.
We show the following monotonicity property:
\begin{observation} \label{obs:Monotonicity}
	For every two sets of agents $S$ and $T$, such that $T\subseteq S$, agent $i \in T$,
	price $p$ and a (possibly fractional) number of items $k$, $\E_\moo[x^T_i] \geq \E_\moo[x^S_i]$ where $x^T, x^S$
	are the allocations returned by $\moo(T,p,k)$ and $\moo(S,p,k)$, respectively.
\end{observation}
\begin{proof}
	Let $B_S$, $B_T$ be the expected sum of the budgets of all agents
	with value at least $p$ that are processed before agent $i$ in a
	random permutation on sets $S$ and $T$ respectively. It is clear
	that $B_T \leq B_S$.
\end{proof}
Budget limitations might prevent the allocation of integral items. 
In offline settings with indivisible items, a common solution for this problem is to sell correlated lottery tickets to the agents. 
Due to the online nature of our settings, this cannot be done. Therefore, We propose to realize an agent's allocation within her time frame, via uncorrelated lotteries, as described in the procedure below.

This procedure calculates the number of items allocated
to each agent, when agents arrive in a random order.
The procedure returns two allocation vectors, $x$ and $\tilde{x}$ where $x$ determines the actual allocation,
and $\tilde{x}$ determines the allocation agents are charged for.
For every agent $i$, the allocation $x_i$ must be integral (since items are indivisible), while $\tilde{x}_i$
might be fractional.

\begin{figure}[H]
	\MyFrame{
		$\mooo \lp S, p, k\rp$
		\begin{enumerate}
			\item Let $\pi$ be a random permutation of $S$.
			\item Initialize $k'\gets k$, $x_i\gets0$ and $\tilde{x}_i\gets 0$ for every $i\in S$.
			\item For $i=1,2,\ldots,|S|$:
			\begin{enumerate}
				\item If $v_{\pi(i)} \geq p$:
				\begin{itemize}
					\item If $k'\leq \frac{b_{\pi(i)}}{p}$:\\
					\mytab Set $x_{\pi(i)}\gets k'$ and $\tilde{x}_{\pi(i)}\gets k'$.
					\item Else:\\
					\mytab Set $\tilde{x}_{\pi(i)}\gets \frac{b_{\pi(i)}}{p}$.\\
					\mytab Set $x_{\pi(i)}\gets\begin{cases}
					\lceil\frac{b_{\pi(i)}}{p}\rceil &\mbox{with probability }\frac{b_{\pi(i)}}{p}-\lfloor\frac{b_{\pi(i)}}{p}\rfloor.\\
					\lfloor\frac{b_{\pi(i)}}{p}\rfloor &\mbox{with probability }1-\frac{b_{\pi(i)}}{p}+\lfloor\frac{b_{\pi(i)}}{p}\rfloor.\\
					\end{cases}$
				\end{itemize}
				\item $k'\gets k'-x_{\pi(i)}$.
			\end{enumerate}
			
			\item Return $(x, \tilde{x})$.
		\end{enumerate}
	}
	\caption{A procedure for the allocation of indivisible items.}
	
	\label{alg:M'}
\end{figure}

The following lemma shows that by using uncorrelated lotteries we lose at most a factor of $2$. The proof is deferred to Appendix \ref{appsec:div_indiv}.
\begin{lemma} \label{lm:rounding}
	Let $(x,\tilde{x})$ be the tuple returned by
	$\mooo\left(S,p,k\right)$ and $(x',\tilde{x}')$ the tuple returned by
	$\moo\left(S,p,k\right)$. Then,
	$\E_{\mooo}\left[\sum_{i\in S}\tilde{x}_i\right] \geq \sum_{i\in
		S}\tilde{x}'_i/2$.
\end{lemma}

\subsection{Online Implementation of Random Sampling Based Mechanisms} \label{sec:rs_base}

In this section we present a high level technique for the implementation of random sampling
mechanisms in online settings. In the sequel, we show that by using suitable allocation and pricing
functions, this technique can be used to devise truthful and nearly optimal mechanisms
in online budgeted settings. However, as discussed in Section \ref{result}, this technique can be used for
other settings, with or without budgets.

For every agent $i$ let $\hat{a}_i, \hat{d}_i$ be the reported arrival time and departure time of agent $i$ respectively.
Let $N$ be a set of agents and $m$ be an amount of items.
Let $\pay$ be a function that receives a set of agents as input and returns a price.
Let $\alloc$ be a function that receives as input a set of agents, a price per item,
and an amount of items, and returns two allocation vectors, $x$ and $\tilde{x}$,
where $x$ is the actual allocation and $\tilde{x}$ is used for charging the agents.
Note that $x_i$ and $\tilde{x}_i$ might be different (e.g., if $\alloc$ runs lotteries).

Mechanism  $\mtempl$ works as follows:
Toss a fair coin $n-1$ times, and let $j$ be the number of tails plus one.
Place the first $j$ agents in set $A$ and the rest in set $B$.
Upon the arrival of the last agent in set $A$, apply the random sampling paradigm on the set $A$.
Upon the arrival of each agent in set $B$, place her into set $B_1$ or $B_2$, by tossing a fair coin.
For each agent placed in $B_1$, give her allocation as if she was in set $A_1$, and charge her a price computed according to set $A_2$.
Apply an analogous allocation and charging for agents in set $B_2$.

\begin{figure}[H]
	\MyFrame{
		$\mtempl\lp N, m, \pay, \alloc \rp$\\\\
		{\bf Partition:}
		\vspace{-3mm}
		\begin{enumerate}
			\item Toss a fair coin $n-1$ times independently. Let $j$ be the number of tails \emph{plus} 1.\\
			Let $t_0 \gets \hat{a}_j$ (the reported arrival time of the $j$-th agent).
			
			\item Let $A$ be the set of the first $j$ agents (based on reported arrival
			time) and let $B$ be $N \setminus A$.
		\end{enumerate}
		{\bf Sampling phase:} (Based on set $A$)
		\vspace{-3mm}
		\begin{enumerate}
			\item At time $t_0$, partition the agents in $A$ into two sets $A_1$ and $A_2$ by
				tossing an independent fair coin for each agent. Assume without loss of generality that agent $j$ belongs
				to $A_2$.
				
			\item Let $(x,\tilde{x})$ be the tuple returned by
				$\alloc(A_1,\pay(A_2),\frac{m}{4})$ and $(x',\tilde{x}')$
				the tuple returned by $\alloc(A_2,\pay(A_1),\frac{m}{4})$.
				
			\item Resolve the allocation and payments of the sampling set:
				
			\begin{enumerate}
					
				\item For every agent $i \in A_1$, if $\hat{d}_i \geq t_0$
					(i.e., agent $i$ did not depart yet), allocate to agent $i$ $x_i$
					items and charge her $\tilde{x}_i \cdot \pay(A_2)$.
					
				\item Apply an analogous procedure for agents in
					$A_2$, i.e., for every agent $i \in A_2$, if $\hat{d}_i \geq t_0$,
					allocate to agent $i$ $x'_i$ items and charge her $\tilde{x}'_i \cdot
					\pay(A_1)$.
					
			\end{enumerate}
		\end{enumerate}			
		{\bf Revenue collection phase:} (Based on set $B$)
		\vspace{-3mm}
		\begin{enumerate}
			\item Let $m_1 \gets \frac{m}{4}$, $m_2 \gets \frac{m}{4}$.
						
			\item For every agent $i \in B$, once she arrives, toss a fair coin.
						
			\begin{enumerate}
				\item If she turns heads: 
				\begin{enumerate}
					\item Assign agent $i$ to set $B_1$.
					\item Let $(x',\tilde{x}')$ be the tuple returned by $\alloc\left(\{A_1 \cup i\},\pay(A_2-j),\frac{m}{4}\right)$.
					\item Let $x_i \gets \min\{x'_i, m_1\}$,$\tilde{x}_i \gets \min\{\tilde{x}'_i, m_1\}$.
					\item  Allocate to agent $i$ $x_i$ items at a price of
					$\tilde{x}_i \cdot \pay(A_2 - j)$.
					\item Update $m_1 \leftarrow m_1 - x_i$.
				\end{enumerate}
				\item If the coin turns tails:
				\begin{enumerate}
					\item Assign agent $i$ to set $B_2$.
					\item Let $(x',\tilde{x}')$ be the tuple returned by $\alloc\left(\{A_2 \cup i\}
					- j,\pay(A_1),\frac{m}{4}\right)$.
					\item Let $x_i \gets \min\{x'_i, m_2\}$,
					$\tilde{x}_i \gets \min\{\tilde{x}'_i, m_2\}$.
					\item Allocate To agent $i$ $x_i$ items at a price of
					$\tilde{x}_i \cdot\pay(A_1)$.
					\item Update $m_2 \leftarrow m_2 - x_i$.
				  \end{enumerate}
				\end{enumerate}
			\end{enumerate}
	}
	\caption{A random sampling based online mechanism.}  \label{alg:TemplSampling}
\end{figure}

Let $N' = N - \{j\}$. The following observation shows that sets $A_1,A_2-j,B_1$ and $B_2$ form a random partition of the agents in $N'$.

\begin{observation} \label{obs:Partition}
	Mechanism $\mtempl$ places each agent $i \in N'$ randomly and independently in either
	sets $A_1$, $A_2 - j$, $B_1$ or $B_2$ with equal probability.
\end{observation} 
\begin{proof}
	Observe that $|N'|=n-1$ and $|A - j| \sim B(n-1,1/2)$. Since the agents of $N'$ arrive in a random order, the
	sets $A-j$ and $B$ form a uniformly-sampled random partition of $N'$. Moreover, agents in
	$A - j$ are divided into sets $A_1$ and $A_2 - j$ uniformly at random, and agents in set $B$ are divided into sets $B_1$
	and $B_2$ in a similar manner.
\end{proof}

\section{Revenue Maximization} \label{sec:rev}
In this Section we study mechanisms for the allocation of identical items (both divisible and indivisible)
with constant approximation to the optimal revenue under a large market assumption.

Let $S$ be a set of agents, $p$ a price and $k$ a
number of items. We define $p^*_{S,k} =\\
\argmax_p\left\{\min\left(\sum_{i \in S, v_i \geq p}{\frac{b_i}{p}},k\right) \cdot p\right\}$, i.e., $p^*_{S,k}$ is
the price that maximizes the revenue of selling at most $k$ items to
agents in $S$ at a uniform price. Let $OPT(S,k) =
\min\left(\sum_{i \in S, v_i \geq
p^*_{S,k}}{\frac{b_i}{p}},k\right) \cdot p^*_{S,k}$, i.e., the optimal
revenue from selling at most $k$ items to agent in $S$ at a uniform
price. Let $OPT^*(S,k)$ be the optimal revenue from selling at most
$k$ items to agents in $S$ at heterogeneous prices. 
The following theorem from
\cite{abrams2006revenue} shows that a mechanism that gives a constant approximation to $OPT(S,k)$ also gives a constant
approximation with respect to $OPT^*(S,k)$.
\begin{theorem} \cite{abrams2006revenue} \label{thm:uniToArbitraryRev}
	For every set of agents $S$ and $k$ items, $OPT(S,k) \geq
	\frac{OPT^*(S,k)}{2}$.
\end{theorem}

\subsection{The Mechanism}
For the allocation of divisible and indivisible items we apply the $\mtempl$ mechanism described in the previous section. 
For divisible items, we define $\modd = \mtempl(N, m, \payrev, \moo)$, 
where $\payrev(S)=p^*_{S,\frac{m}{4}}$, and $\moo$ is as defined in Section \ref{sec:div_indiv}.
For indivisible items, we define 
$\mo = 
\mtempl(N, m, \payrev, \mooo)$, where $\payrev$ is as above, and $\mooo$ is as defined in Section \ref{sec:div_indiv}.

Recall that without a large market assumption, no truthful mechanism can achieve a constant approximation to the revenue.
Similarly to \cite{borgs2005multi}, we characterize the size of the market by $\epsilon$, described as follows:
Let $b^{\max}_{S}= \max\{b_i\}_{i\in S}$ be the maximum budget of
any agent in set $S$; $\epsilon(S,k)=\frac{b^{\max}_{S}}{OPT(S,k)}$.
Intuitively, a smaller $\epsilon$ implies a larger market.
The main result of this section is:
\begin{theorem} \label{thm:RevMain}
$\mo$ is a truthful mechanism that allocates indivisible items and gives a $32$-approximation to the optimal revenue 
as $\epsilon$ tends to 0. 
In addition, $\modd$ is a truthful mechanism that allocates divisible items and gives a $16$-approximation to the optimal revenue 
as $\epsilon$ tends to 0.
\end{theorem}
In what follows we prove the theorem for the allocation of indivisible items. 
It is not difficult to observe that proving $32$-approximation for indivisible items directly implies 
$16$-approximation for divisible items. This is because our analysis goes through the allocation of divisible 
items and then applies Lemma \ref{lm:rounding} to show that the transition to indivisible items loses at most a 
factor $2$ in the approximation. 
\subsection{Analysis}
\begin{theorem} \label{thm:Truthfulness}
	Mechanism $\mo$ is truthful, i.e., each agent maximizes her
	utility by reporting her true type.
\end{theorem}
\begin{proof}
	Let $i$ be any agent. We prove that the utility of agent $i$
	can only decrease if she misreports her type. \\\\
	\textbf{Misreporting $\bm{v_i}$, $\bm{b_i}$:} Agents in set $A$ and $B$ cannot
	change their price or position by misreporting their value or budget.
	Since given a fixed price and a position the mechanism allocates the best possible
	allocation, agents do not have the incentive to misreport their bids.\\\\
	\textbf{Misreporting $\bm{d_i}$:} Agents in $B$ are being allocated upon
	arrival. Hence, misreporting $d_i$ has no effect. Agents in $A$ are
	being allocated at time $t_{0}$. Therefore, by
	reporting an earlier departure time, agent $i$ might depart before
	time $t_{0}$. By delaying her departure time,
	agent $i$ might receive the items outside her time frame.\\\\
	{\bf Misreporting $\bm{a_i}$:} We analyze the following cases.		
	\begin{enumerate}
		\item The mechanism places agent $i$ in set $B$ according to both her true arrival
		time and reported arrival time. Recall that agents in $B$ are being
		allocated upon arrival. By reporting an earlier arrival time, agent
		$i$ receives the items outside her time frame. By delaying her
		arrival time, the agent might receive less items (since other agents might be allocated before agent $i$).
		\item The mechanism places agent $i$ in set $B$ according to her true arrival
		time and in set $A$ according to her reported arrival time.
		Since agent $i$ is placed in set $B$ according to her true arrival time, $a_i > t_{0}$.
		Hence, when agent $i$ is placed in set $A$ she receives the items outside her time frame.
		
		\item The mechanism places agent $i$ in set $A$ according to both her true arrival
		time and reported arrival time. Since agent $i$ is placed in set
		$A$, by misreporting her arrival time she can only affect the
		arrival order and the time $t_{0}$. Observe that the mechanism allocates the agents
		in a random order and takes into consideration
		departed agents (although they do not receive items).
		
		\item The mechanism places agent $i$ in set $A$ according to
		her true arrival time and in set $B$ according to her reported
		arrival time. Note that by moving to set $B$, agent $i$ caused
		agent $j$ to move to set $A$. Let $(x,\tilde{x})$ be the tuple returned by
		$M_1 = \mooo\left(\{A_1 \cup i\},p^*_{A_2 - j,\frac{m}{4}},\frac{m}{4}\right)$. Let
		$(x',\tilde{x}')$ be the tuple returned by	
		$M_2 = \mooo\left(\{A_2 \cup i\} - j,p^*_{A_1,\frac{m}{4}},\frac{m}{4}\right)$. If
		agent $i$ had been reporting her true arrival time, she would have
		been placed in set $A_1$ with probability $\frac{1}{2}$, and in set
		$A_2$ with probability $\frac{1}{2}$.
		
		Note that $\E_{M_1}[x_i]=\E_{M_1}[\tilde{x}_i]$ and $\E_{M_1}[x'_i]=\E_{M_1}[\tilde{x'}_i]$. Therefore, by reporting her true arrival time, with probability
		$\frac{1}{2}$, agent $i$ receives an expected amount of 
		$\E_{M_1}[x_i]$ items at a price of $\E_{M_1}[x_i] \cdot p_{A_2
			- j}$, and with probability $\frac{1}{2}$, agent $i$ receives an expected amount of $\E_{M_2}[x'_i]$ items at a price of
		$\E_{M_2}[x'_i] \cdot p_{A_1}$. Since agent $i$ misreported her
		arrival time, she is placed in set $B_1$ with probability
		$\frac{1}{2}$ and placed in set $B_2$ with probability of
		$\frac{1}{2}$. Hence, with probability $\frac{1}{2}$, agent $i$ receive 
		an expected amount of $x \leq \E_{M_1}[x_i]$ items at an expected
		price of $x \cdot p_{A_2 - j}$, and with probability $\frac{1}{2}$,
		agent $i$ receives an expected amount of $x \leq \E_{M_2}[x'_i]$ items at an expected price of $x \cdot p_{A_1}$ (recall that for set $B$
		the allocation is bounded by the remaining amount of the item). It is clear that the expected
		utility of agent $i$ can only decrease by misreporting her arrival
		time.
	\end{enumerate}
\end{proof}
Recall that $N' = N - j$. Let $OPT^* = OPT^*(N,m)$, $OPT = OPT(N,m)$, $\epsilon =
\epsilon(N,m)$, $OPT' = OPT(N',m)$, $\epsilon' = \epsilon(N',m)$.
The following lemma shows that there is no significant
difference between $OPT'$ and $OPT$:
\begin{lemma} \label{lm:OPTtToOPTRelationRev}
$\Pr\left[OPT'\geq OPT/2 \right]\geq 1-1/n$.
\end{lemma}
\begin{proof}
	There can be at most a single agent that can contribute more then
	$OPT/2$ to the optimal revenue from set $N$. Since agent $j$ is
	chosen uniformly at random from set $N$ (recall that the agents arrive in a random order), 
	the probability that such an agent (if exists) is not in $N - j$ is $1/n$.
\end{proof}
Our analysis uses theorems from \cite{borgs2005multi}. Their mechanism, which we will refer to as $\mof$ is based on the random sampling paradigm. The mechanism divides the agents into two sets $S_1$ and $S_2$ uniformly at random. Then it calculates the optimal uniform price for selling at most $k/2$ items to agents in sets $S_1$ and $S_2$, denoted by $p_1$ and $p_2$ respectively. As a final step, the mechanism sells $k/2$ items to agents in $S_1$ at a price per item $p_2$, and $k/2$ items to agents in $S_2$ at a price per item $p_1$. For a formal description see Figure \ref{fig:offline} at Appendix \ref{appsec:Omit_Mech}.

Let $\overline{\epsilon} = \epsilon(S,k)$.
 Borgs et al. showed the following:
\begin{theorem}\footnote{Note that this is stronger than the statement of Lemma 5.2 in \cite{borgs2005multi},
		but this is essentially what they proved.} 
	\cite{borgs2005multi} \label{thm:highrev}
    Let $r_{S_1}$ be the revenue from selling at most $k/2$ items to set $S_1$ at a price of $p^*_{S,k}$.
    Let $r_{S_2}$ be the revenue from selling at most $k/2$ items to set $S_2$ at a price of $p^*_{S,k}$.
    For any $\delta\in [0,1]$, the probability that both $r_{S_1} \geq \frac{1-\delta}{2}OPT(S,k)$ and
    $r_{S_2} \geq \frac{1-\delta}{2}OPT(S,k)$ is at least $1-2e^{-\delta^2/(4\overline{\epsilon})}$.
\end{theorem}
\begin{theorem}\footnote{Note that this result is shown in proof of Theorem 5.1 in \cite{borgs2005multi}.} \cite{borgs2005multi} \label{thm:budget-offline}
	Let $r_{S_1}$ be the revenue of mechanism $\mof$ from agents in set $S_1$.
	Then, for any $\delta\in [0,\frac{1}{3}]$, the probability that $r_{S_1} \geq \frac{1-3\delta}{2}OPT(S,k)$ is at least $1-2e^{-\delta^2/(4\overline{\epsilon})}$.
\end{theorem}
Mechanism $\mof$ sells items in an arbitrary order, in particular,
according to the online arrival order. Moreover,
according to Observation \ref{obs:Partition} (see Section \ref{sec:rs_base}), our mechanism partition the agents into sets in
an equivalent way to mechanism $\mof$. Therefore, we can use the theorems from \cite{borgs2005multi}
in our analysis. 

Let $N_1 = B_1 \cup A_{2}-j$ and $N_2 = B_2 \cup A_1$.
Let $OPT_1 = OPT(N_1,m/2)$, $\epsilon_1 = \epsilon(N_1,m/2)$, $OPT_2 = OPT(N_2,m/2)$
and $\epsilon_2 = \epsilon(N_2,m/2)$.
The following lemma proves that with high probability both sets receive a significant fraction of $N'$.
\begin{lemma} \label{lm:LargeEps}
$\Pr\left[OPT_1 \geq \frac{1}{4}OPT', OPT_2 \geq \frac{1}{4}OPT', \epsilon_1 \leq 4\epsilon',
\epsilon_2 \leq 4\epsilon'\right] \geq 1-2e^{-1/16\epsilon'}$.
\end{lemma}
\begin{proof}
	Using Observation \ref{obs:Partition}, we get that $N_1$ and $N_2$ form a uniformly sampled partition of $N'$. Hence, we can apply
	Theorem \ref{thm:highrev} with $\delta = \frac{1}{2}$, $S = N'$,
	$S_1 = N_1$, $S_2 = N_2$, $k = m$. Therefore, with probability of
	$1-2e^{-1/16\epsilon'}$ we get that $OPT_1 \geq r_{N_1} \geq
	\frac{1}{4}OPT'$, $OPT_2 \geq r_{N_2} \geq \frac{1}{4}OPT'$. Since
	$b^{\max}_{N_1} \leq b^{\max}_{N'}$ and $b^{\max}_{N_2} \leq
	b^{\max}_{N'}$ we achieve the bounds for  $\epsilon_1$ and
	$\epsilon_2$ stated in the assertion of the lemma \big(recall that
	$\epsilon(S,k)=\frac{b^{\max}_{S}}{OPT(S,k)}$\big).
\end{proof}
Our mechanism, as opposed to offline random sampling based mechanisms, 
calculates the allocation for each agent in set $B_1$ as if she was in a random permutation in set
$A_1$ (and applies an analogous procedure for sets $B_2$ and $A_2 -
j$). This modification is vital for the truthfulness of the
mechanism.

The following lemma is used to show that we lose only a constant
factor due to this modification. Let $x$ be the allocation returned by $\moo(S,p,k)$.
$\mooa(S,p,k)$ (see Figure \ref{fig:M1} at Appendix \ref{appsec:Omit_Mech}) 
is a mechanism that allocates the items
according to $x$ and charges a price of $p$ per item. $\moob(S,p,k,z)$  (see Figure \ref{fig:M2} at Appendix \ref{appsec:Omit_Mech}) is a
mechanism that sells according to the modification described above,
i.e., it divides set $S$ into sets $S_1$ and $S_2$ such that
each agent has a probability of $z$ to be placed in set $S_1$ and a
probability of $1-z$ to be placed in set $S_2$. Then it sells the
items to agents in set $S_1$ in a random order at a fixed price $p$ per item.
The allocation for each agent is calculated as if the agent was in a
random permutation in set $S_2$.
\begin{lemma} \label{lm:Restriction}
	For every set $S$, price $p$, fraction of item $k$ and a probability
	$z$, the expected revenue from mechanism $\moob(S,p,k,z)$ is at least
	$z$ times the expected revenue from mechanism $\mooa(S,p,k)$.
\end{lemma}
\begin{proof}
	Let $r$ be the expected revenue from mechanism $\mooa(S,p,k)$. For the
	sake of the analysis, assume that mechanism $\mooa(S,p,k)$ begins by
	placing each agent $i \in S$ in $S_1$ and $S_2$ with respective
	probability $z$ and $1-z$ (as mechanism $\moob(S,p,k,z)$
	does). Since each agent has a probability of $z$ to be placed in
	$S_1$, the expected revenue extracted from agents in $S_1$ is $z
	\cdot r$. Therefore, it is sufficient to show that for each partition
	of $S$ into $S_1$ and $S_2$, the expected revenue from $\moob(S,p,k,z)$
	equals to at least the expected revenue of $\mooa(S,p,k)$ from agents in set
	$S_1$. There are two cases. If $\moob(S,p,k,z)$ allocated the entire
	fraction of the item, then it collected the maximal revenue possible
	from price $p$ and the claim is obviously true.
	
	Otherwise, for each agent $i \in S_1$, $\min\{x'_i, m\} = x'_i$ at
	step \ref{alg: M2_Allocation} of mechanism $\moob$. We prove that for
	each such an agent $i$, her expected revenue in $\moob(S,p,k)$ (denoted
	by $r_{i,2}$) equal at least her expected revenue in $\mooa(S,p,k)$
	(denoted by $r_{i,1}$). Since $\min\{x'_i, m\} = x'_i$ then $r_{i,2}
	= \E_\moo[\bar{x}_i] \cdot p$ where $\bar{x}$ is the allocation returned
	by $\moo(\{S_2 \cup i\},p,k)$. $r_{i,1} = \E_\moo[\tilde{x}_i] \cdot p$
	where $\tilde{x}$ is the allocation returned by $\moo(S,p,k)$. Since
	$\{S_2 \cup i\}\subseteq S$, then according Observation \ref{obs:Monotonicity}
	$\E_\moo[\bar{x}_i] \geq \E_\moo[\tilde{x}_i]$. This complete the proof.
\end{proof}
We are now ready to prove the main technical lemma of this section.

\begin{lemma} \label{thm:OPTtConstantRev}
For every $\delta\in [0,\frac{1}{3}]$, the probability that the expected revenue obtained from mechanism $\mo(N,m)$
is greater than $\frac{(1-3\delta)OPT'}{16}$ is at least\\
$\left(1-2e^{-1/16\epsilon'}\right)\left(1-4e^{-\delta^2/(16\epsilon')}\right)$.
\end{lemma}
\begin{proof}
	We calculate the expected revenue only for the event that $OPT_1
	\geq \frac{1}{4}OPT'$ and $OPT_2 \geq \frac{1}{4}OPT'$ (and
	therefore also $\epsilon_1 \leq 4\epsilon'$, $\epsilon_2 \leq
	4\epsilon'$). Let $r$ be the expected revenue from mechanism
	$\mooa(B_1,p^*_{A_2 - j,\frac{m}{4}},\frac{m}{4})$. Note that each
	agent $i \in S_1$ is randomly and independently placed to either $B_1$
	or $A_2 - j$ with equal probability. Hence, we can apply Theorem
	\ref{thm:budget-offline} with $S_1 = B_1$, $S_2 = A_2 - j$ and get
	that the probability that $r \geq \frac{1-3\delta}{2}OPT_1 \geq
	\frac{1-3\delta}{8}OPT'$ is at least $1-2e^{-\delta^2/(4\epsilon_1)}
	\geq 1-2e^{-\delta^2/(16\epsilon')}$. Note that the expected revenue
	from mechanism $\mooa(B_1 \cup A_1,p^*_{A_2 -
		j,\frac{m}{4}},\frac{m}{4})$ is at least the expected revenue from
	mechanism $\mooa(B_1,p^*_{A_2 - j,\frac{m}{4}},\frac{m}{4})$, since
	the agents in set $B_1 \cup A_1$ with value at least $p^*_{A_2 -
		j,\frac{m}{4}}$ are a superset of the ones in set $B_1$. Hence, the
	probability that the expected revenue from mechanism $\mooa(B_1 \cup
	A_1,p^*_{A_2 - j,\frac{m}{4}},\frac{m}{4})$ is at least
	$\frac{1-3\delta}{8}OPT'$ is at least
	$1-2e^{-\delta^2/(16\epsilon')}$. 
	
	By applying Lemma
	\ref{lm:Restriction} with $S = B_1 \cup A_1$, $S_1 = B_1$, $S_2 =
	A_1$, $p = p^*_{A_2 - j,\frac{m}{4}}$, $k = \frac{m}{4}$ and $z =
	\frac{1}{2}$ we get the following. The probability that the expected
	revenue from $\moob(B_1 \cup A_1,p^*_{A_2 -
		j,\frac{m}{4}},\frac{m}{4}, \frac{1}{2})$ is greater than
	$\frac{1}{2}\cdot\frac{1-3\delta}{8}OPT' = \frac{1-3\delta}{16}OPT'$ is at least
	$1-2e^{-\delta^2/(16\epsilon')}$.
	By a symmetric argument we get that the
	probability that the expected revenue from mechanism $\moob(B_2 \cup
	A_2 - j,p^*_{A_1,\frac{m}{4}},\frac{m}{4}, \frac{1}{2})$ is greater than
	$\frac{1-3\delta}{16}OPT'$ is at least
	$1-2e^{-\delta^2/(16\epsilon')}$. Using a union bound we get that
	the probability that the expected revenue from mechanism $\moob(B1
	\cup A_1,p^*_{A_2 - j,\frac{m}{4}},\frac{m}{4}, \frac{1}{2})$ and
	mechanism $\moob(B_2 \cup A_2 - j,p^*_{A_1,\frac{m}{4}},\frac{m}{4},
	\frac{1}{2})$ is greater than $\frac{1-3\delta}{8}OPT'$ is at least $1-4e^{-\delta^2/(16\epsilon')}$. 
	Note that we can use Lemma \ref{lm:Restriction} although we enforce that
	$OPT_1 \geq \frac{1}{4}OPT'$ and $OPT_2 \geq \frac{1}{4}OPT'$
	(which implies restrictions on $N_1 = B_1 \cup A_{2 - j}$ and $N_2 = B_2 \cup A_1$), 
	since this conditions are symmetric.	
	
	One can easily observe that the expected revenue from mechanism $\modd(N, m)$ is at 
	least the expected revenue from mechanism $\moob(B1 \cup
	A_1,p^*_{A_2 - j,\frac{m}{4}},\frac{m}{4}, \frac{1}{2})$ and
	mechanism $\moob(B_2 \cup A_2 - j,p^*_{A_1,\frac{m}{4}},\frac{m}{4},
	\frac{1}{2})$. To bound the expected revenue from mechanism $\mo(N, m)$
	we apply Lemma \ref{lm:rounding} and lose an additional factor of $2$ in the revenue.
	Combining the above with Lemma \ref{lm:LargeEps}, we get the desired result.
\end{proof}
We now use Theorem \ref{thm:uniToArbitraryRev} and Lemma \ref{lm:OPTtToOPTRelationRev},  to bound the expected
revenue from mechanism $\mo$ with respect to $OPT^*$.

\begin{theorem} \label{coro:ConstantRev}
 For every $\delta\in [0,\frac{1}{3}]$, the probability that the expected revenue obtained from mechanism $\mo(N,m)$
 is greater than $\frac{(1-3\delta)OPT^*}{64}$ is at least\\
 $(1-1/n)\left(1-2e^{-1/32\epsilon}\right)\left(1-4e^{-\delta^2/(32\epsilon)}\right)$.
\end{theorem}
\begin{proof}
	According to Lemma \ref{lm:OPTtToOPTRelationRev}, $\Pr\left[OPT'\geq
	OPT/2 \right]\geq 1-1/n$. Combining  $OPT' \geq OPT/2$ and
	$\epsilon' \leq 2\epsilon$ with the bound from Lemma
	\ref{thm:OPTtConstantRev} and Theorem \ref{thm:uniToArbitraryRev}
	yields the desire result.
\end{proof}
We our now ready to show the main result of this section.\\\\
\textbf{Proof of Theorem \ref{thm:RevMain}:}
By definition, when $\epsilon$ tends to $0$, $OPT'$ approaches to $OPT$. 
Therefore, the bound from
Lemma \ref{thm:OPTtConstantRev} becomes: for every $\delta\in
[0,\frac{1}{3}]$, the probability that the expected revenue obtained from
mechanism $\mo(N,m)$ is greater than $\frac{(1-3\delta)OPT^*}{32}$ is
at least
$\left(1-2e^{-1/32\epsilon}\right)\left(1-4e^{-\delta^2/(32\epsilon)}\right)$.
Since $\epsilon$ tends to $0$, for every $\delta$,
$\left(1-2e^{-1/32\epsilon}\right)\left(1-4e^{-\delta^2/(32\epsilon)}\right)$
tends to $1$.\qedsymb\\\\

\section{Liquid Welfare Maximization} \label{sec:liquid}

In the liquid welfare setting the seller has a single divisible item\footnote{By normalization, this is equivalent to $m$ divisible items.}.
Let $\tilde{v}_i\lp x_i\rp=\min\lp
v_ix_i, b_i\rp$. The liquid welfare from an allocation vector 
$x$ is $\tilde{W}\lp x\rp=\sum_i{\tilde{v}_i\lp
x_i\rp}$. Let $OPT$ denote the optimal liquid welfare. 
For an agent $i \in N$ we
define $\bar{v}_i = \min\{v_i,b_i\}$ to be the \emph{liquid value}
of the agent. We refer to an agent whose liquid value is a constant
fraction of $OPT$ as a ``dominant" agent.

Our mechanism, which we refer to as $\mool$, is composed of two mechanisms.
It runs mechanism $\moov(N)$ with probability $\mu$ and mechanism
$\moos$ with probability $1-\mu$, where the exact value of $\mu$ will be determined in the analysis.

We define $\moos = \mtempl(N, 1, \bar{P}, \moo)$,
where $\moo$ is as defined in Section \ref{sec:div_indiv}, and $\bar{P}$ is described as follows: 
Let $S$ be a set of $n'$ agents such that $v_{i_1} \geq v_{i_2} ... \geq
v_{i_{n'}}$. Let $k$ be the maximum integer such that $\sum_{j=1}^{k}
b_{i_j} \leq v_{i_{k}}$. The function $\bar{P}(S)$ is the market clearing price\footnote{
	Market clearing price is a term taken from the general
	equilibrium theory. In the special case of a songle type of good, 
 market clearing price is the maximum price satisfying
	that each agent can be assigned an optimal fraction of the good such
	that there is no deficiency or surplus.} defined as
follows: $\bar{P}(S) = \max\{\sum_{j=1}^{k}
b_{i_j},v_{i_{k+1}}\}$.

Mechanism $\moov(N)$, which is defined in Figure \ref{fig:mvcg}, achieves a constant approximation to the optimal liquid welfare if there is a single dominant agent, and
Mechanism $\moos$ achieves a constant approximation in all other cases. 
The main result of this section is:
\begin{theorem}\label{thm:liq_main}
Mechanism $\mool$ is a truthful mechanism for the allocation of a divisible item that gives a constant approximation to the optimal liquid welfare.
\end{theorem}
\begin{figure}[H]
	\MyFrame{
		$\moov\lp N\rp$
		\begin{enumerate}
			\item Let $A$ be the set of the first $\frac{n}{2}$ agents (based on reported arrival
			time) and let $B$ be $N \setminus A$.
			Let $t_{0.5} = \hat{a}_{\frac{n}{2}}$ (the reported arrival time of the $\frac{n}{2}$-th agent).
			
			\item At time $t_{0.5}$: Order the agents in set $A$ according to their
			liquid value such that $\bar{v}_{i_1} \geq \bar{v}_{i_2} ... \geq \bar{v}_{i_{\frac{n}{2}}}$. Set $p_2=\gamma \cdot\bar{v}_{i_2}$ and $p_1=\gamma \cdot\bar{v}_{i_1}$. Toss a fair coin.
			\begin{itemize}
				\item If it turns heads, then if $\bar{v}_{i_1} \geq p_2$ and $d_{i_1} \geq t_{0.5}$, allocate the entire item to agent $i_1$ at a price $p_2$.
				\item If the coin turns tail, then wait for the first agent $j\in B$ such that $\bar{v}_{j} \geq p_1$ (if exists)
				and allocate to her the entire item at a price $p_1$.
			\end{itemize}
		\end{enumerate}
	}
	\caption{Online modified VCG mechanism.} \label{fig:mvcg}
\end{figure}

The value of $\gamma > 1$ is determined in the analysis.

\subsection{Analysis}

\begin{theorem} \label{thm:liq_truth}
The mechanism is truthful for any $\gamma > 1$, i.e., each agent maximizes her utility by reporting her true type.
\end{theorem}

The proof of the above theorem is deferred to Appendix \ref{appsec:liquid_analysis}

We now bound the performance guarantee of the mechanism.
Some parts of the analysis resemble the analysis from
~\cite{lu2014liquid}. We bound the performance guarantee for several
cases. For some of them we prefer to bound the expected revenue
instead of bounding the expected liquid welfare. The relation
between revenue and liquid welfare is stated in Lemma $3.4$ from
~\cite{lu2014liquid}.

\begin{lemma} \label{lm:RevenueLiquidRelation} ~\cite{lu2014liquid}
	The liquid welfare produced by any truthful and budget feasible
	mechanism is at least the revenue of the auctioneer.
\end{lemma}

We introduce the following definitions. Recall that $j$ is the last agent in set $A$ (according to her reported arrival time), where $j\sim
B(n-1,1/2)+1$. Let $N' = N - j$. Let $p^* = \bar{P}(N')$ be the
market clearing price of the total set of agents without agent $j$.
Let $OPT$ and $OPT'$ be the optimal liquid welfare for sets $N$ and $N'$
respectively. Let $V = \{i \in N|v_i \geq p^*\}$ and $V' = \{i \in
N'|v_i \geq p^*\}$. For every set of agents $T \subset N'$ we define
$B_T = \sum_{i \in T} b_{i}$, $B'_T = \sum_{i \in T \cap V'} b_{i}$.
By definition, $B_V' \geq p^*$ (Recall that the market clearing
price is the maximal price such that there is no deficiency or
surplus of any good. Hence, it is clear that the budget of the agent
from set $V'$ is at least $p^*$).

The following lemma proves some relations regarding the market clearing price.
In particular, it proves that the market clearing price is a good
approximation for the optimal liquid welfare. This Lemma was stated and proved by Lu and Xiao in an earlier version of \cite{lu2014liquid}, and removed in the latest version. We state and prove it here for the sake of completeness.
\begin{lemma} [From a previous version of ~\cite{lu2014liquid}\footnote{The previous version can be found in \url{http://arxiv.org/pdf/1407.8325v2.pdf}.}]\label{lm:ClearingPriceRelation}
	Let $S$ and $T\subseteq S$ be two sets of agents and let
	$W$ be the optimal liquid welfare from the set $S$. Then:
	\begin{enumerate}
		\item $\bar{P}(S) \geq \frac{1}{2}W$.
		\item $\bar{P}(T) \leq \bar{P}(S)$.
	\end{enumerate}
\end{lemma}
\begin{proof}
	Let $S=\{1,\ldots, n\}$, $v_1\geq v_2\geq \ldots v_n$ and let $k$ be the maximum integer such that $\sum_{i=1}^k b_i\leq v_k$. By the definition of the market clearing price we have that $\bar{P}(S)=\max\{\sum_{i=1}^k b_i,v_{k+1}\}$.
	Let $x^*= \{x^*_i\}_{i\in S}$ be an allocation that maximizes the social welfare of the agents in $S$. First we prove the first part of the lemma: $$W=\sum_{i\in S}\min\{v_ix^*_i, b_i\}\leq \sum _{i=1}^k b_i+\sum_{i=k+1}^n v_ix^*_i\leq \sum _{i=1}^k b_i+v_{k+1}\sum_{i=k+1}^n x^*_i\leq \sum _{i=1}^k b_i+v_{k+1}\leq 2\bar{P}(S),$$ where the first inequality is derived from the definition of the $\min$ function, the third inequality is due to the fact that only one items is allocated and the fourth inequality results from the definition of clearing market price.
	
	Now we prove the second part of the lemma. Let $T=\{i_1,i_2,\ldots,i_{n'}\}$ such that $v_{i_1}\geq v_{i_2}\geq \ldots v_{i_{n'}}$. Let $k'$ be the maximum $i\in\{1,\ldots, k\}$ such that $k'\in T$.
	Let $k_T$ the maximum integer such that $\sum_{j=1}^{k_T} b_{i_j}\leq v_{i_{k_T}}$ (i.e., $\bar{P}(T)=\max\{\sum_{i=1}^{k_{T}} b_i,v_{k_T+1}\}$). 
	Since $k' \leq k$, according to the definition of $k$, $\sum_{j=1}^{k'} b_{j}\leq v_{k'}$. 
	Combining with the fact that T is a subset of S we have that either $i_{k_T}=k'$ or $i_{k_T}> k$. 
	For the first case we have that $\bar{P}(T)= \max\{\sum_{j=1}^{k_T} b_{i_j},v_{i_{k_T+1}}\}\leq \max\{\sum_{j=1}^{k} b_{j},v_{k+1}\}= \bar{P}(S)$ since $i_{k_T} = k' \leq k$ and $i_{k_T+1}>k$ (according to the definition of $k'$). 
	On the other hand, if $i_{k_T}> k$ then $\bar{P}(T)=\max\{\sum_{j=1}^{k_T} b_{i_j},v_{i_{k_T+1}}\}\leq  v_{i_{k_T}}\leq v_{k+1}\leq \max\{\sum_{j=1}^{k} b_{j},v_{k+1}\}= \bar{P}(S)$.
	When the first inequality results from the definition of the clearing market price.
	 
	This complete the proof of the lemma.
\end{proof}
As a corollary of the Lemma \ref{lm:ClearingPriceRelation}, we get that
$p^* \geq \frac{1}{2}OPT'$, $p_{A_1} \leq p^*$ and $p_{A_2 - j}
\leq p^*$. The following lemma shows that there is not a significant
difference between $OPT'$ and $OPT$. 
\begin{lemma} \label{lm:OPTtToOPTRelation}
	$\Pr\left[OPT'\geq \big(1-\frac{1}{k+1}\big)OPT \right]\geq 1-\frac{k}{n}$.
\end{lemma}
\begin{proof}
	Let $x^*$ be an optimal allocation. Let $I_k$ be the set of the $k$ agents with the highest liquid welfare under $x^*$. 
	Note that for every agent $\ell\in (N\setminus I_k)$, we have that $\tilde{v}_{\ell}(x^*_{\ell})\leq \frac{OPT}{k+1}$; otherwise, 
	$\sum_{i\in (I_k+\ell)}\tilde{v}_i(x^*_i)> (k+1)\frac{OPT}{k+1} = OPT$ 
	for some $\ell\in (N\setminus I_k)$. 
	Since agent $j$ is chosen uniformly at random from set $N$ (recall that the agents arrive in a random order), 
	the probability that $j\in I_k$ is $\frac{k}{n}$. In case that $j\in (N\setminus I_k)$, 
	we get that $$OPT'\geq \sum_{i\in N-j}\tilde{v}_i(x^*_j)\geq OPT-\frac{OPT}{k+1} = \left(1-\frac{1}{k+1}\right)OPT.$$
\end{proof}
Using Observation \ref{obs:Partition}, we get that every agent of $V'$ is in
each of the sets $A_1$, $A_2-j$, $B_1$ and $B_2$ with an equal and
independent probability. Let $Q = {A_1 \cup B_1 \cup B_2}$. Hence, each
agent of $V'$ is placed in set $Q$ with a probability of $\frac{3}{4}$ and in set $A_2-j$ with a probability of $\frac{1}{4}$.

For simplicity, we normalize $p^* = 1$. The
following lemma shows that if all the agents in set $V'$ have a
relatively small budget, both $Q$ and $A_2-j$ get a significant
fraction of the budget of the agents in set $V'$.

\begin{lemma} \label{lm:DistributeAgents}
	If $\max_{i \in V'} b_i \leq \frac{\gamma}{10}$ then $\Pr\left[\min\{B'_{Q}, B'_{A_2- j}\} \geq \frac{1}{15}\right] \geq 1 - \frac{135}{242}\gamma$.
\end{lemma}
\begin{proof}
	Recall that $B_{V'} \geq p^* = 1$. Assume that $B_{V'} = 1$ (if $B_{V'} > 1$, we can decrease some of the budgets, and the claim
	only gets stronger). Let $\mathbb{I}_i$ be a random indicator such that $\mathbb{I}_i = 1$
	if agent $i$ belongs to set $Q$ and $\mathbb{I}_i = 0$ otherwise. Therefore,
	$B'_{Q} = \sum_{i \in V'} b_{i}\mathbb{I}_i$, $\E\left[B'_Q\right] = \frac{3}{4}$ and
	\begin{eqnarray*}
		Var\left[B'_Q\right] & = & \sum_{i \in V'} Var\left[b_{i}\mathbb{I}_{i}\right]   =  \sum_{i \in V'}\left(\E\left[(b_{i}\mathbb{I}_{i})^2\right] - \E\left[b_{i}\mathbb{I}_{i}\right]^2\right)
		\\ & = & \sum_{i \in V'}\frac{3}{16}b_{i}^2 \leq \frac{1}{(\gamma /10)}
		\frac{3}{16} \cdot (\gamma /10)^2
		= \frac{3}{160} \gamma,
	\end{eqnarray*}
	where the inequality holds since setting $b_i=\gamma/10$ for every $i\in V'$ maximizes $\sum_{i \in V'}b_{i}^2$ under the constraints $\max_{i \in V'} b_i\leq \gamma/10$ and $\sum_{i\in V'}b_i=1$.
	Using Chebyshev's inequality we get:
	$$\Pr\left[\left| B'_{Q} - \frac{3}{4}\right| \geq 11/60\right] \leq \frac{Var\left[B'_{Q}\right]}{\left(\frac{11}{60}\right)^2} \leq \frac{135}{242}\gamma.$$
	When $|B'_{Q} - \frac{3}{4}| \leq \frac{11}{60}$ then
	$\frac{17}{30} \leq B'_{Q} \leq \frac{14}{15}$. Since $B'_{Q} + B'_{A_2 - j} = B_V' = 1$, we get that
	$B'_{A_2 - j} \geq \frac{1}{15}$, which completes the proof.
\end{proof}

In the following theorem we uses the relation between mechanisms
$\mooa$ and $\moob$ (see Figures \ref{fig:M1} and \ref{fig:M2}) presented in Section \ref{sec:rev}. Recall
that $\mooa(S,p,k)$ is a mechanism that allocates the item
according to $\moo(S,p,k)$ and charges a price of $p$ per the entire item.
$\moob(S,p,k,z)$ is a mechanism that sells according to the
following modification. It divides a set $S$ into sets $S_1$ and
$S_2$ such that each agent has a probability of $z$ to be placed in
set $S_1$ and a probability of $1-z$ to be placed in set $S_2$.
It then sells the items to agents in set $S_1$ in a random order at a
price $p$ per the entire item. The allocation for each agent is calculated as if
the agent was in a random permutation in set $S_2$. 

Note that for $n<100$, devising an $100$-approximation truthful mechanism is trivial --- allocate the entire item to a randomly chosen agent. Therefore, from now on we assume that $n\geq 100$. Now we are ready to prove our main result of this section:
\begin{theorem}
	Setting $\mu = 1/10$ and $\gamma=1.0001$, mechanism $\mool$ gives a $2443$-approximation to the optimal liquid welfare.
\end{theorem}
\begin{proof}
We bound the expected welfare for the event that $OPT'\geq \frac{9}{10}OPT$ occurs. 
In this case $OPT \leq \frac{10}{9}OPT'\leq \frac{20}{9}p^*=\frac{20}{9}$, where he second inequality stems from Lemma \ref{lm:ClearingPriceRelation}.

We bound the performance of the mechanism for three different cases. Assume by renaming that $1$ and $2$ are the agents with the highest liquid value and the second highest liquid value in set $N$ respectively.\\\\
\textbf{Case A} (no ``dominant" agent): $\max_{i\in V}b_i \leq \gamma/10$.\\
We bound the expected revenue from mechanism $\moos(Q,p_{A_2-j},\frac{1}{4})$; by Lemma \ref{lm:RevenueLiquidRelation} it is sufficient.
It is clear that $\max_{i\in V'}b_i \leq \gamma/10$. We calculate the expected revenue for the event that $\min\{B'_Q,B'_{A_2-j}\}\geq 1/15$ occurs. In this case, by Lemma \ref{lm:ClearingPriceRelation}, $p_{A_2-j}=\bar{P}(A_2-j)\geq B'_{A_2-j}/2\geq \frac{1}{30}$. First we calculate the expected revenue from $\mooa(Q,p_{A_2-j},\frac{1}{4})$. There are two possible outcomes. If the entire $\frac{1}{4}$ fraction was sold, then the revenue is at least $\frac{1}{4\cdot 30} = \frac{1}{120}$. 
Otherwise, the entire budget of the agents with value at least $p_{A_2 - j}$ was
exhausted. According to Lemma \ref{lm:ClearingPriceRelation},
$p_{A_2 - j} \leq p^*$. Therefore, the budget of agents in set $Q
\cap V'$ was exhausted and the revenue is at least $B'_Q\geq \frac{1}{15}$. 
Combining the two cases we get that the expected revenue from mechanism $\moos(Q,p_{A_2-j},\frac{1}{4})$ is at least 
min$\{\frac{1}{120},\frac{1}{15}\} = \frac{1}{120}$.
Applying Lemma \ref{lm:Restriction} with $S=Q$ and $z=1/3$ we get that the expected revenue from mechanism $\moob(Q,p_{A_2-j},\frac{1}{4},\frac{1}{3})$ is at least $\frac{1}{3\cdot 120} = \frac{1}{360}$.

Observe that
the expected revenue obtained from agents in set $B_1$ by running
mechanism $\moos(N)$ is at least the expected revenue from $\moob(Q,
p_{A_2 - j},\frac{1}{4}, \frac{1}{3})$, since when running mechanism
$\moos(N)$ the agents in set $B_1$ are competing only against agents
from set $A_1$ and when running $\moob(Q, p_{A_2 - j},\frac{1}{4},
\frac{1}{3})$ they have to compete against agents from set $A_1 \cup
B_2$.
We get that the expected revenue is at least $(1-\mu)(1-\frac{145}{242}\gamma)/360\geq 1/1000$.\\\\ 
\textbf{Case B} (More than one ``dominant" agent): $\max_{i\in V}b_i > \gamma/10$ and $\bar{v}_1<\gamma\bar{v}_2$.\\
We bound the expected revenue from mechanism $\moos$. Since $\max_{i \in V} b_i \geq \gamma/10$, $\bar{v}_{1} \geq \min\{1,
\gamma /10\} = \gamma /10$ (recall that $p^*$ is normalized to $1$). Moreover, $\bar{v}_2>\bar{v}_1/\gamma\geq 1/10$.
Let $i'$ be the agent with the highest liquid value in set $V$. Clearly, $\bar{v}_{i'}\geq \gamma/10$. 

Next, we bound the probability that agent $j$ is not one of the agents $i',1,2$. Recall that the number of agents is at least $100$.
We assume that $OPT'\geq \frac{9}{10}OPT$, and therefore, $j$ is not one of the 9 agents with the highest welfare in the optimal allocation.
It is clear that the worst case is when agents $i',1,2$ are not among the 9 agents with the highest welfare. Hence, this probability
is at least $\frac{88}{91}$.

Let $Q_1=A_1\cup B_1$ and $Q_2=A_{2}-j\cup B_2$. Note that at least one of the agents $1, 2$ is not agent $i'$. Assume without lost of generality that agent $2$ is different than agent $i'$ (the other case is similar).

We inspect the case where either $2\in A_1\wedge i'\in Q_2$ or $2\in A_{2}-j\wedge i'\in Q_1$. This happens with probability $1/4$
(recall that each agent in set $N'$ has the same independent probability of being placed in either $A_1$,  $B_1$, $A_{2}-j$ or $B_2$).

Assume without loss of generality that $2\in A_1$ and $i'\in Q_2$ (the other case is symmetric). 
Using the same argument as in the previous section, we have that
the expected revenue from mechanism $\mooa(Q_2,p_{A_1},1/4)$ is at least min$\{\frac{1}{80},\frac{\gamma}{10}\} = \frac{1}{80}$.
By using Lemma \ref{lm:Restriction} with $S=Q_2$ and $z=1/2$ we get that the expected revenue out of agents in $B_2$ in $\moos$ is at least $\frac{1}{160}$. Therefore, the expected revenue in this case is at least $(1-\mu)\frac{88}{91\cdot 4\cdot 160}\geq \frac{13}{10000}$. \\\\
\textbf{Case C} (one ``dominant" agent): $\max_{i\in V}b_i > \gamma/10$ and $\bar{v}_1<\gamma\bar{v}_2$.\\
We bound the expected liquid welfare from mechanism $\moov$.
With probability $1/4$ agent 1 is placed in set $B$ and agent $2$ is placed in set $A$, and with probability $1/2$ the item is sold in the second phase of mechanism $\moov$. Hence, the expected liquid welfare is at least 
$\frac{\gamma}{2\cdot4\cdot10}\mu= \frac{1}{80}\mu\gamma = \frac{1250125}{1000000000}$.\\

We conclude that the expected liquid welfare is at least $min\{\frac{1}{1000}, \frac{13}{10000}, \frac{1250125}{1000000}\} = \frac{1}{1000}$.
Recall that in the event that $OPT'\geq \frac{9}{10}OPT$, $OPT \leq \frac{20}{9}$.  
According to Lemma \ref{lm:OPTtToOPTRelation}, $\Pr[OPT'\geq \frac{9}{10}OPT]\geq 91/100$.
Therefore, the approximation ratio is at most $\frac{\frac{20}{9}}{\frac{1}{1000}\frac{91}{100}} < 2443$. 
\end{proof}

\bibliographystyle{abbrv}
\bibliography{tswbBBL}

\appendix
\section{Omitted Mechanisms} \label{appsec:Omit_Mech}
\begin{figure}[H]
\MyFrame{
$\mof\lp S, k\rp$
\begin{enumerate}
    \item Partition the agents in set $S$ into two sets $S_1$ and $S_2$,
          by tossing an independent fair coin for each agent. Set $p_1\gets p^*_{S_1, m/2}$ and $p_2\gets  p^*_{S_2, m/2}$.

    \item Let $x$ be the allocation returned by $\moo(S_1, p_2, k/2)$
          and $x'$ be the allocation returned by $\moo(S_2, p_1, k/2)$.

    \item For every agent $i \in S_1$: Allocate to agent $i$ $x_i$ items at a price of $x_i \cdot p_2$.

    \item For every agent $i \in S_2$: Allocate to agent $i$ $x'_i$ items at a price of $x'_i \cdot p_1$.

    \item Run a correlated lottery between the agents who received fractions of items to resolve the real allocation.
\end{enumerate}
} \caption{The offline mechanism presented in \cite{borgs2005multi} for selling indivisible items and maximizing revenue.}
\label{fig:offline}
\end{figure}

\begin{figure}[H]
	\MyFrame{
		$\mooa\lp S, p, k\rp$
		\begin{enumerate}
			\item Let $x$ be the allocation returned by $\moo\left(S,p,k\right)$.
			
			\item For each agent $i \in S$ allocate $x_i$ items at price of $x_i \cdot p$.
			
		\end{enumerate}
	} \caption{A mechanism for selling divisible items at a given price.}
	\label{fig:M1}
\end{figure}

\begin{figure}[H]
	\MyFrame{
		$\moob\lp S, p, k, z\rp$
		\begin{enumerate}
			\item Initialize $m \gets k$.
			
			\item For each agent $i \in S$: With probability of $z$ place $i$ in set $S_1$ and with probability of $1-z$ in set $S_2$.
			\item For each agent $i\in S_1$ (for an arbitrary order of the agents):
			\begin{enumerate}
				\item  Let $x'$ be the allocation returned by $\moo\left(\{S_2 \cup i\},p,k\right)$. Let $x_i \gets\min\{x'_i, m\}$. Allocate to agent $i$ $x_i$ items at a price of $x_i \cdot p$.  \label{alg: M2_Allocation}
				
				\item Update $m \gets m - x_i$.
			\end{enumerate}
			
		\end{enumerate}
	} \caption{A mechanism defined solely for the analysis of our mechanisms.}
	\label{fig:M2}
\end{figure}

\section{Proof of Lemma \ref{lm:rounding}} \label{appsec:div_indiv}
We consider only the case where each agent in $S$ has a budget
smaller than $p$ (it is clear that for the integral part of
$\frac{b_i}{p}$, mechanisms $\mooo$ and $\moo$ will allocate
the same amount of items). Let $S' = \{i \in S:v_i \geq p\}$ (agents
in $S \setminus S'$ have no effect on the allocation). Let $B_{S'} =
\sum_{i\in S'}b_i$. We analyze the following cases:
\\\\	
\textbf{Case A:}	
$k = 1$ and $B_{S'} \geq p$.
It is clear that in this case $\sum_{i\in S}\tilde{x}'_i = 1$ (mechanism
$\moo$ sells the entire item).
$\E_\mooo\left[\sum_{i\in S}\tilde{x}_i\right]$ is the expected
amount of items charged by $\mooo$ before the item
is sold. Let $j$ be the size of set $S$ and assume by renaming that $1,\ldots,j$ is the order in which
mechanism $\mooo$ iterates the agents in set $S'$.
Every time an agent $i\in S'$ is accessed before the item is sold,
she is charged for a $\tilde{x}_i = b_i/p$ fraction, and gets the item
with probability $\tilde{x}_i$. 

Let $P_i$ be the probability that agent $i\in
S'$ is accessed before the item is sold. Clearly,
$X=\sum_{i=1}^{j}P_i\cdot \tilde{x}_i$ is the expected number of
items the agents of $S'$ are charged for before the item is sold.

For the sake of the analysis, consider the following. 
Let $r=\min\{t\in S': \sum_{i=1}^t \tilde{x}_i \geq 1\}$.
Set $y_i=\tilde{x}_i$ for every $i < r$, $y_{r}=1-\sum_{i=1}^{r-1}y_i$ 
and $y_i = 0$ for every $i>r$.
Let $z\sim U[0,1]$ be a uniformly sampled number between 0
and 1. Let $\tilde{j}=\min\{t\in S': \sum_{i=1}^t y_i \geq z\}$, and
let $X'=\sum_{i=1}^{\tilde{j}} y_i$. Since $\E[z]=1/2$, we have
that $\E[X']\geq 1/2$. Let $P'_i$ be the probability that $i\leq
\tilde{j}$. Therefore, $\E[X']=\sum_{i=1}^{j} P'_i\cdot y_i$.

To complete this case we show by induction that $P_i \geq P'_i$
for every $i \in S'$. Clearly $P_1 = P'_1 = 1$. For $i > 1$, by definition
$P_i=P_{i-1}\cdot (1-\tilde{x}_{i-1})$. On the other hand, 
$$P'_i = P'_{i-1}\cdot\left(1-\frac{y_{i-1}}{\sum_{t=i-1}^{j}y_t}\right)\leq P_{i-1}\cdot(1-y_{i-1})= 
P_{i-1}\cdot(1-\tilde{x}_{i-1})= P_i,$$ 
where the inequality stems from the induction hypothesis.
Since $\tilde{x}_i\geq y_i$ for any $i \in S'$, we conclude that $X\geq \E[X'] \geq 1/2$.\\\\
\textbf{Case B:}
$k \geq 1$ and $B_{S'} \leq p$ (agents can afford at most a single item). 
By a similar analysis to that of the \textbf{Case A}, we get that
$\E_\mooo\left[\sum_{i\in S}\tilde{x}_i\right] \geq \sum_{i\in S}\tilde{x}'_i/2$.
\\\\
\textbf{Case C:}
$k \geq 2$ and $\sum_{i\in S}\tilde{x}'_i = k' > 1$. 
For simplicity we assume that $k'$ is integral, but a similar argument can
be applied to the fractional case. Recall that mechanism $\mooo$ sells the items
sequentially. Let $B_{S',\ell}$ be the sum of the budgets of the agents in set $S'$, 
after mechanism $\mooo$ sells item $\ell-1$ (for $\ell=1$ clearly $B_{S',\ell} = B_{S'}$).
We consider the following cases:

\begin{enumerate}
	\item Mechanism $\mooo$ sold at least $k'-1$ items and $B_{S',k-1} \geq p$. Therefore, before selling each of the items $1,\ldots, k'$, we have that the sum of the budgets of agent in $S'$ is at least $p$. Hence, we can apply the analysis of \textbf{Case A} and conclude that
	$\E_{\mooo}\left[\sum_{i\in S}\tilde{x}_i\right] \geq k'/2$.
	
	\item Mechanism $\mooo$ sold less then $k'-1$ items or $B_{S',k-1} < p$. In this case, agents in $S$ were charged for at least $k'-1$ items (i.e. $\sum_{i\in S}\tilde{x}_i \geq k'-1$). 
	Since $k' \geq 2$, $k'-1 \geq k'/2$. This complete the proof.
\end{enumerate}\qedsymb\\\\
\section{Proof of Theorem \ref{thm:liq_truth}} \label{appsec:liquid_analysis}
Since the mechanism is a probabilistic combination of two mechanism, it is sufficient to show that each of the mechanisms is truthful.\\\\
{\bf $\moos$ Mechanism}\\
One can easily verify that mechanism $\moos$ is almost identical to
mechanism $\mo$. The only differences between the mechanisms are
the prices and the allocation method which does not use lotteries. Hence, the proof is similar to the proof of Theorem \ref{thm:Truthfulness}.\\\\
{\bf $\moov$ Mechanism}\\
Most of the cases have a similar proof as in mechanism  $\mo$. We describe only the cases with a different proof. Let
$i$ be any agent. We prove that the utility of agent $i$ can only
decrease if she misreports her type.\\\\
\textbf{Misreporting ${\bm v_i},{\bm b_i}$:} Agents in set $B$ cannot change their price
or position by misreporting their value or budget. Since given a fixed
price and a position the mechanism allocates the best possible
allocation, agents do not have the incentive to misreport their bids.

For agent $i\in A$, if $\bar{v}_i \geq \gamma \cdot \bar{v}_{i_2}$
then the agent has a probability of $\frac{1}{2}$ to receive the
item at a price of $\gamma \cdot \bar{v}_{i_2}$. Agent $i$ cannot
affect the price or the probability to win the item by misreporting
$\bar{v}_i$. If $\bar{v}_i < \gamma \cdot \bar{v}_{i_2}$, then to
win the item agent $i$ has to bid at least $\gamma \cdot
\bar{v}_{i_2}$ and to exceed her budget or value.\\\\
\textbf{Misreporting ${\bm a_i}$:} We analyze the following case (analogous to case 4 of misreporting $a_i$ in $\mo$).\\ 
The mechanism places agent $i$ in set $A$ according to her true arrival
time and in set $B$ according to her reported arrival time. Let
$\bar{v}_{i_1} \geq \bar{v}_{i_2} ... \geq
\bar{v}_{i_{\frac{n}{2}}}$ be the agents in set $A$ ordered
according to their liquid value, assuming that agent $i$ reported her
real arrival time. If $\bar{v}_i < \gamma \cdot \bar{v}_{i_2}$,
agent $i$ will not receive the item whether she is in set $A$ or in set $B$.
If $\bar{v_i} \geq \gamma \cdot \bar{v}_{i_2}$, then the agent
receives the item at price $\gamma \cdot \bar{v}_{i_2}$ with
probability $1/2$ whether she is in set $A$ or in set $B$  (note that in this case, if
agent $i$ is placed in set $B$ then agent $i_2$ becomes the agent
with the highest liquid value in set $A$).

\section{Tie Breaking} \label{appsec:tie-break}

In this section we show how to perform tie breaking when an agents is allowed to report the same arrival time as another agent, but cannot report an earlier arrival time than her real arrival time. The tie breaking rule we present is as follows --- whenever an agent $i$ arrives at time $a_i$, the mechanism chooses a uniformly at random value $\tilde{a}_i\sim [0,1]$ and sets agent $i$'s arrival time to be $\langle a_i,\tilde{a}_i\rangle$. The mechanism orders the agents according to the following rule: Agent $i$ precedes agent $j$ if $a_i< a_j$ or if $a_i=a_j$ and $\tilde{a}_i< \tilde{a}_j$. 

We change mechanisms $\mo$, $\modd$, $\moov$ and $\moos$ according to the above rule. 
We claim that these mechanisms are truthful. 
The only problematic case is when an agent reports an earlier arrival time such that she is placed in set $A$ instead of being placed is set
$B$ (the proof all all the other cases is similar to the proof of Theorem \ref{thm:Truthfulness}).
Since an agent cannot affect the random value assigned to her, and since we specifically prevent an agent from reporting an earlier arrival time, this is impossible.

\end{document}